\documentclass[12pt]{article}
\usepackage{amsmath,amsthm,amssymb,amsbsy}

\usepackage[pdftex]{graphicx}

\newtheorem{thm}{Theorem}[section]

\newtheorem{prop}[thm]{Proposition}
\theoremstyle{definition}
\newtheorem{defn}[thm]{Definition}
\theoremstyle{remark}

\theoremstyle{example}
\newtheorem{xmpl}[thm]{Example}
\theoremstyle{assumption}
\newtheorem{asmp}[thm]{Assumption}

\newcommand{\R}{{\mathbb{R}}}

\begin{document}
\title{On Bankruptcy Game Theoretic Interval Rules}
\author{ Rodica Branzei \\
  University ``Alexandru Ioan Cuza'', Ia\c{s}i, Romania
  \\
  \texttt{branzeir@info.uaic.ro} \and Marco Dall'Aglio \\
Luiss University, Rome, Italy \\
 \texttt{mdallaglio@luiss.it} \and
Stef Tijs \\ Tilburg University, The
  Netherlands \\ \texttt{S.H.Tijs@uvt.nl}}

\date{November 22, 2012}

\maketitle

\begin{abstract}
Interval bankruptcy problems arise in situations where an estate
has to be liquidated among a fixed number of creditors and
uncertainty about the amounts of the claims is
modeled by intervals. We extend in the interval setting the
classical results by Curiel, Maschler and Tijs (1987)
that characterize division rules which correspond to solutions of
the cooperative bankruptcy game. Finally, we analyze the difficulties with incorporating the uncertainty about the estate.


\end{abstract}

\section{Introduction}

Bankruptcy problems provide a simple and effective mathematical
model to describe situations where an estate has to be divided
among a fixed number of individuals (creditors or players) who
advance claims with total value too large to be fully compensated
by the value of the estate. The foundations for these models are
set in the works of O'Neill \cite{o82} and Aumann and Maschler
\cite{am85}. These authors analyze the seemingly mysterious
solutions for specific instances of a bankruptcy problem
prescribed in the Babylonian Talmud and find that the answers
given by the ancient book are in fact solutions of a cooperative
game, called the bankruptcy game, played by the creditors. Curiel,
Maschler and Tijs \cite{cmt87} consolidate the links between
bankruptcy problems and cooperative game theory by studying the
whole class of division rules for bankruptcy problems which
correspond to solutions of the corresponding bankruptcy game. They
provide a characterization of such rules  by means of a truncation
property: the solution based on the bankruptcy game are those, and
only those, that ignore claims which are higher than the whole
estate, and reduce them to the value of the estate. The same work
also characterizes division rules that provide allocations
belonging to the core of the bankruptcy game. A review of the literature on bankrupt\-cy problems is given in Moulin \cite{m02} and Thomson \cite{t03}.

The bankruptcy problem studied in those pioneering works requires
an exact knowledge of all the terms of a bankruptcy problem. We
allow instead for a certain degree of uncertainty on the problem
data. In fact, claimants often face uncertainty regarding
their effective rights and, as a result, individual claims can be expressed in
the form of closed intervals without any probability distributions attached
to them. In such situations our model based on interval claims fits better
than the more standard claims approach and, additionally, offers flexibility in
conflict resolution under interval uncertainty of the estate at stake. Economic
applications of our approach include funds’ allocation of a firm among its
divisions (Pulido, S\`anchez-Soriano and Llorca \cite{psl02}, Pulido et al. \cite{pbhls08}),
priority problems Moulin \cite{m00}, distribution of penalty costs in delayed
projects (Branzei et al.\ \cite{bfft02}) and disputes related to cooperation in joint
projects where agents have restricted willingness to pay (Tijs and Branzei
\cite{tb04}).

Our aim is to extend the general result by Curiel, Maschler and
Tijs regarding bankruptcy problems with classical (or exact) data
to the interval setting. Can we characterize interval division
rules which correspond to solutions for interval bankruptcy games?
Special care is placed on the definitions of the entities and of
the operations in the new environment. In particular, we 
verify that the characterization Theorem of
Curiel, Maschler and Tijs (Theorem 5 in \cite{cmt87}) can be extended to the case of interval claims. However, we show through counterexamples that a similar extension to the case 
where the  interval uncertainty affects both the claims and the estate is not so straightforward.

The remainder of the paper is organized as follows: Section 2
reviews the definitions and the results of interest in the
classical setting; Section 3 introduces interval bankruptcy
problems and their truncations; Section 4 provides the extension of 
the characterization Theorem with interval uncertainty regarding the claims and uncertainty-free estate.
In the last section we show the difficulties of including interval uncertainty for the estate.

\section{The classical setting} All the results in this section are
taken from Curiel, Maschler and
Tijs \cite{cmt87}. We consider situations where a monetary
{\em estate} $E \in \mathbb{R}_+$ has to be divided among a set of
{\em claimants} (players) $N = \{1,2,\ldots,n\}$, each having a
{\em claim} $d_i$, $i \in N$, over the estate. Let
$d=(d_1,d_2,\ldots,d_n) \in \mathbb{R}^n_+$ denote the $n$-tuple
of claims. It is assumed that $E \leq \sum_{i \in N} d_i$. A {\em division rule} $f$
associates with each bankruptcy problem $(E,d)$ an $n$-tuple
$f(E,d) = (f_1(E,d),f_2(E,d),\ldots,f_n(E,d))$ such that
\[
d_i \geq f_i(E,d) \geq 0 \mbox{ for every }i \in N \quad
\mbox{and} \quad \sum_{i \in N}f_i(E,d) = E.
\]
A cooperative game, called the {\em bankruptcy game}, is defined
by 
$$
v_{E,d}(S)=(E - \sum_{i \in N \setminus S} d_i)_+ \: ,
$$
where $(x)_+= \max \{0,x\}.$ $v_{E,d}(S)$
denotes the minimal amount that the coalition $S \subset N$ will
receive, once the claims of the creditors outside $S$ have been
fully compensated. Individual claims larger than the whole estate
may be considered unreasonable. If this is the case, we consider
the {\em truncated bankruptcy problem} $(E,d \wedge E)$,
where $d \wedge E = (d_1 \wedge E,d_2 \wedge E,\ldots,d_n \wedge
E)$ and
$a \wedge b = \min \{a,b\}$.

\begin{defn}
A division rule $f$ for a bankruptcy problem is a game theoretic
division rule if there is a solution concept $g$\footnote{In \cite{cmt87} the notion of solution concept is not defined explicitly. We will simply assume that $g$ is any mapping from $\R_+^{2^n - 1}$ to $\R_+^n$. Any game $v$ among $n$ players is fully defined by $2^n -1$ values since $v(\emptyset)=0$ always holds.} for cooperative
games such that
\[
f(E,d) = g(v_{E,d}) \mbox{ for every bankruptcy problem} (E,d) \in BR^N.
\]
\end{defn}
We will focus on extension of the following Theorem.
\begin{thm}\label{thm5}(\cite{cmt87}, Theorem 5) A division rule
$f$ for bankruptcy problems is a game theoretic division rule if
and only if $f(E,d)=f(E,d \wedge E)$.
\end{thm}

\section{The interval setting} Here we extend our analysis to the
situation where the claim of each player is expressed
as an interval (instead of a single value). To cope with interval uncertainty, we denote by $I(\R_+)$ the set of all closed and bounded intervals in $\R_+$, and by $I(\R_+)^N$ the set of all  $n-$dimensional vectors whose elements belong to $I(\R_+)$. Let $I,J \in I(\R_+)$ with $I=[\underline{I},\overline{I}], J=[\underline{J},\overline{J}],$ then $I+J=[\underline{I} + \underline{J},\overline{I} + \overline{J}]$ and $I - J = [\underline{I} - \overline{J},\overline{I} - \underline{J}].$ We say that $I$ is {\em weakly better} than $J$, which we denote by $I \succeq J,$ or, equivalently by $J \preceq I,$ if $\underline{I} \geq \underline{J}$ and $\overline{I} \geq \overline{J}.$ For further reference on interval mathematics we refer to Moore \cite{m95}.

Here we deal
with a ``crisp''  estate $E$ and
an $n$-tuple  of interval claims
$[d]=([\underline{d}_1,\overline{d}_1],[\underline{d}_2,\overline{d}_2]
,\ldots,[\underline{d}_n,\overline{d}_n])$. We will assume that
\begin{equation} \label{basicassmp}E \leq \sum_{i \in N}
\underline{d}_i .
\end{equation}
meaning that in no case the estate will be sufficiently rich to satisfy the claimants' requests in full.

\begin{defn}
An interval bankruptcy rule determines, for each estate
$E$ and each set of interval claims $[d]$ an $n$-tuple of
interval rewards
\[
\mathcal{F}(E,[d])=
(\mathcal{F}_1(E,[d]),\mathcal{F}_2(E,[d]),\ldots,\mathcal{F}_n(E,[d]))
\in I(\R)^N
\]
which are {\em reasonable,} i.e. $[0,0] \preceq \mathcal{F}_i(E,[d])
 \preceq [\underline{d}_i,\overline{d}_i]$ for each $i \in N$, and 
{\em weakly efficient\footnote{As opposed to {\em (strong) efficiency}, which requires $\sum_{i \in N} \mathcal{F}_i(E,[d]) = [E,E].$},} i.e. $\sum_{i \in N} \mathcal{F}_i(E,[d]) \ni
E$.
\end{defn}

We will consider reasonable monotonicity
assumptions for our bankrupt\-cy rules. Denote with $d_{-i}$ the set of all claims but the
claim of the $i$-th player, i.e.\ 
$$
d_{-i}=
\{d_1,\ldots,d_{i-1},d_{i+1},\ldots,d_n\}.
$$
 Our rule $f$ will
satisfy the following
\begin{asmp} \label{asmpf}
We will consider bankruptcy rules satisfying monotonicity in the resources and in both self-regarding and other-regarding claims. More in detail, for every $i \in N$, the component $f_i$ of the classical
bankruptcy rule $f$ is nondecreasing in $E$ and $d_i$, while
it is nonincreasing in each $d_j$, $j \in N \setminus \{i\}$.
\end{asmp}

It can be shown that the most important bankruptcy
rules\footnote{these include: the proportional rule, defined by
$PROP_{i}(E,d)=\frac{d_{i}}{D}E$ for all $i\in N$, with $D =
\sum_{i \in N} d_i$; the constrained equal awards rule, defined by
$CEA_{i}(E,d)=\min\left\{d_{i},\alpha\right\}$, where $\alpha$ is
determined by $\sum_{i\in N} CEA_{i}(E,d)=E$; the constrained
equal losses rule, defined by
$CEL_{i}(E,d)=\max\left\{d_{i}-\beta,0\right\}$, where $\beta$ is
determined by $\sum_{i\in N} CEL_{i}(E,d)=E$; the Talmudic rule
$TAL(E,d)$, given by $CEA(E,\frac{d}{2})$ if $E \leq \frac{D}{2}$
and by $\frac{d}{2} + CEL(E - \frac{D}{2},\frac{d}{2})$
otherwise.} verify Assumption \ref{asmpf} (this is illustrated in
Branzei et al.\ \cite{bdt08}, Appendix A)

\begin{defn} Given a bankruptcy rule $f$, we define the interval bankruptcy rule based on $f$
as
\[ \mathcal{F}(f;E,[d]) =
(\mathcal{F}_i(f;E,[d]))_{i \in N}
\]
where
\begin{equation}
\label{regf} \mathcal{F}_i(f;E,[d]) =
[f_i(E,\underline{d}_i,\overline{d}_{-i}),
f_i(E,\overline{d}_i,\underline{d}_{-i})] \qquad
\mbox{for every }i \in N \; .
\end{equation}
\end{defn}
We say that an interval is {\em tight with respect to a given property} if
each proper subset of that interval {\em does not} satisfy the
same property.
\begin{prop}
\label{prop_tightwe}
Suppose $f$ satisfies Assumption \ref{asmpf}. Then
\begin{enumerate}
\item for each $i \in N$ and $d \in [d]$, we have
\[
f_i(E,d) \in \mathcal{F}_i(f;E,[d])
\]
and, for all $i \in N$, the interval $\mathcal{F}_i(f;E,[d])$ is
tight with respect to selection inclusion.
\item $\mathcal{F}(f,\cdot,\cdot)$ is weakly efficient and reasonable.
\end{enumerate}
\end{prop}
\begin{proof}
To prove {\em (i)} consider the following chain of inequalities
valid by Assumption \ref{asmpf} for each $i \in N$ and $d \in [d]$:
\[
f_i(E,\underline{d}_i,\overline{d}_{-i}) \leq
f_i(E,d) \leq
f_i(E,\overline{d}_i,\underline{d}_{-i}) \; .
\]
Since the extremes are attained, they define the smallest interval
with this property.

To prove {\em (ii)} simply note that the classical bankruptcy rule
$f$ is reasonable, and therefore, for each $i \in N$,
\[
\underline{d}_i \geq
f_i(E,\underline{d}_i,\overline{d}_{-i}) \geq 0 \quad
; \quad \overline{d}_i \geq
f_i(E,\overline{d}_i,\underline{d}_{-i}) \geq 0,
\]
and efficient, so
\begin{gather*}
\sum_{i \in N} f_i(E,\underline{d}_i,\overline{d}_{-i}) 
\leq \sum_{i \in N} f_i(E,\underline{d}_i,\underline{d}_{-i}) = E  \: ;
\\
\sum_{i \in N}
f_i(E,\overline{d}_i,\underline{d}_{-i}) 
\geq \sum_{i \in N} f_i(E,\overline{d}_i,\overline{d}_{-i}) = E \: .
\end{gather*}
\end{proof}

Next, we focus on truncation properties for interval claims. Any
claim that exceeds the estate $E$ may
be considered excessive. Accordingly, we truncate all claims with
respect to this single value. Denote
\[
[d] \wedge E = \left(
[\underline{d}_i \wedge E,\overline{d}_i \wedge
E] \right)_{i \in N}
\]
\begin{defn}
The {\em truncated} interval bankruptcy rule based on $f$ is given
by $\label{truncf}
\mathcal{F}(f;E,[d] \wedge E)$ .
\end{defn}
The truncated interval rule plays an important
role when the underlying classical division rule $f$ is game
theoretic.

\begin{prop} \label{int_crispgtrule}
Suppose that $f$ is a game theoretic division rule satisfying
Assumption \ref{asmpf}. Then
\begin{enumerate}
\item The interval bankruptcy rule coincides with its truncated form, i.e.
\begin{equation}
\label{bankf_eq_trunc} \mathcal{F}(f;E,[d] \wedge E) =
\mathcal{F}(f;E,[d]) \; ;
\end{equation}
\item For each $i \in N$ and $d^* \in [d]$, we have
\[
f_i(E,d^*) \in \mathcal{F}_i(f;E,[d] \wedge E)
\]
and, for all $i \in N$, the interval $\mathcal{F}_i(f,E,[d] \wedge E)$
is tight with respect to selection inclusion;
\item The truncated interval bankruptcy rule is weakly efficient and reasonable.
\end{enumerate}
\end{prop}
\begin{proof}
To prove {\em (i)}, note that, since $f$ is game theoretic
\begin{multline*}
\mathcal{F}(f;E,[d]) = \left(
[f_i(E,\underline{d}_i,\overline{d}_{-i}),
f_i(E,\overline{d}_i,\underline{d}_{-i})] \right)_{i
\in N} =
\\
 \left(
[f_i(E,\underline{d}_i \wedge
E,\overline{d}_{-i} \wedge E),
f_i(E,\overline{d}_i \wedge
E,\underline{d}_{-i} \wedge E)] \right)_{i
\in N} =
\\
 \left(
[f_i(E,\underline{d_i \wedge
E},\overline{d_{-i} \wedge E}),
f_i(E,\overline{d_i \wedge
E},\underline{d_{-i} \wedge E})] \right)_{i
\in N} =
\mathcal{F}(f;E,[d] \wedge E) .
\end{multline*}

To show {\em (ii)} consider the following chain of (in)equalities
valid for each $i \in N$ and $d^* \in
[d]$:
\begin{multline}
\label{incl_trunc}  
f_i(E,\underline{d}_i \wedge
E,\overline{d}_{-i} \wedge E) =
f_i(E,\underline{d}_i,\overline{d}_{-i})  \leq
 \\
  f_i(E,d^*) \leq
f_i(E,\overline{d}_i,\underline{d}_{-i}) =
f_i(E,\overline{d}_i \wedge E,
\underline{d}_{-i} \wedge E) \; .
\end{multline}
The inequalities in the middle derive from the claim monotonicity
properties of $f$, while the fact that $f$ is a game theoretic rule
and Theorem \ref{thm5} explain the equality signs at the extremes.
Once again, the extremes are attained.

Regarding {\em (iii)}, we have
\begin{gather*}
\underline{d}_i \geq
f_i(E,\underline{d}_i,\overline{d}_{-i}) =
f_i(E,\underline{d}_i \wedge
E,\overline{d}_{-i} \wedge E)\geq 0 \; , \\
\overline{d}_i \geq
f_i(E,\overline{d}_i,\underline{d}_{-i}) =
f_i(E,\overline{d}_i \wedge E,
\underline{d}_{-i} \wedge E) \geq 0 \; ,
\end{gather*}
while
\begin{gather*}
\sum_{i \in
N} f_i(E,\underline{d}_i \wedge E,
\overline{d}_{-i} \wedge E) = \sum_{i \in N}
f_i(E,\underline{d}_i,\overline{d}_{-i}) \leq
E \; ,
\\
\sum_{i \in N} f_i(E,\overline{d}_i \wedge
E, \underline{d}_{-i} \wedge E ) = \sum_{i
\in N} f_i(E,\overline{d}_i, \underline{d}_{-i}) \geq
E \; .
\end{gather*}
\end{proof}

\section{Interval Bankruptcy Games and Game Theoretic Rules}
We now extend the notion of bankruptcy game to the specific
interval setting by considering the interval bankruptcy game already defined in 
\cite{bdt03}.
\begin{defn}
The {\em interval bankruptcy game} for the estate $E$
and interval claims $[d]$ is defined, for each $S \subset N$, by
\begin{multline}
\label{defibg} w_{E,[d]}(S) =
[v_{E,\overline{d}}(S),v_{E,\underline{d}}(S)]
=
\\
\left[ \left( E - \sum_{i \in N \setminus S}
\overline{d}_i \right)_+ , \left( E - \sum_{i \in N
\setminus S} \underline{d}_i \right)_+\right] \; .
\end{multline}
\end{defn}
For each $S \subset N$, the interval is delimited by what is left
to coalition $S$ in the worst and in the best possible situation,
respectively, after the players outside $S$ have been compensated
with their full claim.

We now show that every classical bankruptcy game
originating from the estate $E$ and claims $[d]$ is a selection
of the interval bankruptcy game, i.e.\ its values fall in the range of the interval bankruptcy game, and that the each interval game determination is tight with respect to its selections.
\begin{prop} \label{propigt}
For each $S \subset N$ and each
$d^* \in [d]$ 
\[ 
v_{E,d^*}(S) \in
w_{E,[d]}(S) \qquad \mbox{for each }S \subset N \; ,
\]
and each interval $w_{E,[d]}(S)$ is tight.
\end{prop}
\begin{proof}
Simply note that, for each $S \subset N$
and each $d^* \in [d]$,
\[
E - \sum_{i \in N \setminus S} \overline{d}_i \leq
E - \sum_{i \in N \setminus S} d^*_i \leq
E - \sum_{i \in N \setminus S} \underline{d}_i ,
\]
and the chain of inequalities remains valid if we apply the $(\cdot)_+$ operator. Therefore
\[
\underline{w}_{E,[d]}(S) \leq v_{E,d}(S) \leq
\overline{w}_{E,[d]}(S) .
\]
\end{proof}

Just as in the previous section, we consider a truncated form.
\begin{defn}
The {\em truncated interval bankruptcy game}
 is defined, for each $S \subset N$, as $
w_{E,[d] \wedge E}(S).$
\end{defn}
As in the previous case, the two games coincide.
\begin{prop} \label{uegames_coincide}
Every interval bankruptcy game coincides with its truncated form, i.e. 
$$
w_{E,[d] \wedge E}(S)=
w_{E,[d]}(S)\qquad  \mbox{ for all } \quad S \subset N.
$$
\end{prop}
\begin{proof}
It follows from the following equation, valid for any $d^* \in [d]$ and any $S \subset N$
\begin{equation}
\label{eq_tue}
\left( E - \sum_{i \in N
\setminus S} (d^*_i \wedge E)  \right)_+ =
\left( E - \sum_{i \in N \setminus S} d^*_i
\right)_+ \: .
\end{equation}
To prove it, we distinguish two cases: ($i$) if $d^*_i < E$ for any $i \in N \setminus S$ then $d^*_i \wedge E = d^*_i$ for any $i$. ($ii$) If $d^*_j \geq E$ for some $j \in N \setminus S$ then both quantities in \eqref{eq_tue} reduce to 0.
\end{proof}
We now consider those division rules which can be related with
solutions of the interval bankruptcy game defined by \eqref{defibg}.
\begin{defn}
An interval division rule $\mathcal{F}$ is an {\em interval
game theoretic rule} if there exists an interval solution concept
$\mathcal{G}$ for interval cooperative games, i.e.\ a mapping from $I(\R_+)^{2^n - 1}$ to $I(\R_+)^n$, such that
\[
\mathcal{F}(E,[d]) = \mathcal{G}(w_{E,[d]}) .
\]
\end{defn}
We are now able to state a full extension of Theorem \ref{thm5}.

\begin{thm}
\label{int-thm5} An interval division rule $\mathcal{F}$ based on the classical
bankruptcy rule $f$ is game theoretic if and only if the rule
coincides with its truncated form, i.e. \eqref{bankf_eq_trunc}
holds.
\end{thm}
\begin{proof}
Suppose that $\mathcal{F}(f; \cdot, \cdot)$ is game theoretic.
Then, by Proposition \ref{uegames_coincide}, for any estate $E$ and claims $[d]$, 
$$
\mathcal{F}(f;E,[d]) = \mathcal{G}(w_{E,[d]}) =
\mathcal{G}(w_{E,[d] \wedge E}) =
\mathcal{F}(f;E,[d] \wedge E)
$$
and, therefore, \eqref{bankf_eq_trunc} holds.

Conversely, let $f$ be the classical bankruptcy rule on which the rule $\mathcal{F}$ is based, and suppose that \eqref{bankf_eq_trunc} holds. Consider the following interval solution concept $\mathcal{G}^f$
defined for any interval game $w$ 
\[
\mathcal{G}^f (w) = \left( [ f_i(\underline{w}(N),\underline{K}_i^w, \overline{K}_{-i}^w),
f_i(\overline{w}(N),\overline{K}_i^w, \underline{K}_{-i}^w)] \right)_{i \in N} \; ,
\]
where, for each $i \in N$,
\begin{gather*}
\overline{K}_i^w = \overline{M}_i^w + \frac{\left( \overline{w}(N) 
- \sum_{i \in N} \overline{M}_i^w \right)_+}{n}
\\
\underline{K}_i^w = \underline{M}_i^w + \frac{\left( \underline{w}(N) 
- \sum_{i \in N} \underline{M}_i^w \right)_+}{n}
\end{gather*}
and, for each $i \in N$, 
\begin{gather*}
\overline{M}_i^w = \overline{w}(N) - \underline{w}(N \setminus \{ i \}) 
\\
\underline{M}_i^w = \underline{w}(N) - \overline{w}(N \setminus \{ i \}) 
\end{gather*}
Now we apply the interval rule $\mathcal{G}^f$ to the interval bankruptcy game $w_{[E],d}$ defined in \eqref{defibg}. Clearly, $\underline{w}_{[E],d}(N)=\overline{w}_{[E],d}(N)=E$ and, for each $i \in N$
\begin{gather*}
\overline{K}_i^w = \overline{M}_i^w = E - (E - \overline{d}_i)_+ = \overline{d}_i \wedge E \: ,
\\
\underline{K}_i^w = \underline{M}_i^w = E - (E - \underline{d}_i)_+ = \underline{d}_i \wedge E \: ,
\end{gather*}
because
$$
\left( \overline{w}_{[E],d}(N) 
- \sum_{i \in N} \overline{M}_i^w \right)_+ = 
\left( \underline{w}_{[E],d}(N) 
- \sum_{i \in N} \underline{M}_i^w \right)_+= 0 \: .
$$
Thus
$$
\mathcal{G}^f(w_{E,[d]}) = \mathcal{F}(f;E, [d] \wedge E) = \mathcal{F}(f;E, [d]) \: ,
$$
the last equality following from \eqref{bankf_eq_trunc}.
\end{proof}
The following result becomes a direct consequence of the previous Theorem.
\begin{thm}
\label{thm_crisp_int}
Let $f$ be a classical game theoretic bankruptcy rule. Then, the interval bankruptcy rule $\mathcal{F}(f;\cdot,\cdot)$ based on $f$ is also game theoretic.
\end{thm}
\begin{proof}
If $f$ is game theoretic, then \eqref{bankf_eq_trunc} holds by
Proposition \ref{int_crispgtrule}. Apply Theorem \ref{int-thm5}.
\end{proof}

\section{The trouble with the interval estate}
Just as for the claims, the effective amount of the estate may be a source of uncertainty. It is therefore natural to extend the model of bankruptcy interval rule and bankruptcy interval game to the case where both the estate and the claims are of the interval type, i.e.\ $[E] \in I(\R)$ and $[d] \in I(\R)^N$, respectively.

\begin{defn} 
Let $f$ be a classical bankruptcy rule. The interval bankruptcy rule based on $f$ for the interval estate $[E]$ and the interval claims $[d]$ is
as
\[ \mathcal{F}(f;[E],[d]) =
(\mathcal{F}_i(f;[E],[d]))_{i \in N}
\]
where
\begin{equation}
\label{regf} \mathcal{F}_i(f;[E],[d]) =
[f_i(\underline{E},\underline{d}_i,\overline{d}_{-i}),
f_i(\overline{E},\overline{d}_i,\underline{d}_{-i})] \qquad
\mbox{for every }i \in N \; .
\end{equation}
\end{defn}
We have
\begin{prop}
Suppose $f$ satisfies Assumption \ref{asmpf}. Then
\begin{enumerate}
\item for each $i \in N$,  $E^* \in [E]$ and $d^* \in [d]$, we have
\[
f_i(E^*,d^*) \in \mathcal{F}_i(f;[E],[d])
\]
and, for all $i \in N$, the interval $\mathcal{F}_i(f;[E],[d])$ is
tight.
\item $\mathcal{F}(f,\cdot,\cdot)$ is weakly efficient and reasonable.
\end{enumerate}
\end{prop}
The proof is very similar to that of Proposition \ref{prop_tightwe} and is therefore omitted.
We turn our attention to a suitable interval game for this situation.
\begin{defn}
The {\em interval bankruptcy game} for the interval estate $[E]$
and interval claims $[d]$ is defined, for each $S \subset N$, by
\begin{multline}
\label{defibest} w_{[E],[d]}(S) =
[v_{\underline{E},\overline{d}}(S),v_{\overline{E},\underline{d}}(S)]
=
\\
\left[ \left( \underline{E} - \sum_{i \in N \setminus S}
\overline{d}_i \right)_+ , \left( \overline{E} - \sum_{i \in N
\setminus S} \underline{d}_i \right)_+\right] \; .
\end{multline}
\end{defn}
An analogue of Proposition \ref{propigt} holds (where, again, the proof is omitted).
\begin{prop} \label{propigt_intest}
For each $S \subset N$, each $E^* \in [E]$ and each
$d^* \in [d]$ 
\[ 
v_{E^*,d^*}(S) \in
w_{[E],[d]}(S) \qquad \mbox{for each }S \subset N \; ,
\]
and each interval $w_{[E],[d]}(S)$ is tight.
\end{prop}

Having extended the notion of bankruptcy game to the interval setting for both the estate and the claims, we may hope that an
analogue of Theorem \ref{thm_crisp_int} would hold. The
following counterexample, however, highlights a situation where
two instances of interval data pertaining to the same game theoretic classical bankruptcy rule generate the same interval game with two distinct interval rules. The interval bankruptcy rule cannot be derived as an interval solution concept of the interval bankruptcy game.

\begin{xmpl} \label{ctrxmpl1}
Compare the following two situations with two claimants.
\begin{description}
\item[Situation a] $[\underline{E},\overline{E}]_a = [6,8]$,
$[\underline{d}_1,\overline{d}_1]_a=[6,7]$ and
$[\underline{d}_2,\overline{d}_2]_a=[2,3]$
\item[Situation b] $[\underline{E},\overline{E}]_b = [6,8]$,
$[\underline{d}_1,\overline{d}_1]_b=[6,7.5]$ and
$[\underline{d}_2,\overline{d}_2]_b=[2,3]$
\end{description}

If we consider $f$ to be the game theoretic Talmudic (TAL) rule\footnote{The fact that this is a game theoretic rule is shown in Aumann and Maschler \cite{am85}.}, it is easy to verify that 
\begin{multline}
\label{sitsab_differ} \mathcal{F}(TAL;[E]_a,[d]_a) =
([4.5,6.5],[1,2.5]) \neq \\ ([4.5,6.75],[1,2.5]) =
\mathcal{F}(TAL;[E]_b,[d]_b) .
\end{multline}
On the other hand, the two interval games $w_{[E]_a,[d]_a}$ and
$w_{[E]_b,[d]_b}$ coincide, since
\begin{gather*}
w_{[E]_a,[d]_a}(\{1\}) = w_{[E]_b,[d]_b}(\{1\}) = [3,6] \; ;
\\
w_{[E]_a,[d]_a}(\{2\}) = w_{[E]_b,[d]_b}(\{2\}) = [0,2] \; ;
\\
w_{[E]_a,[d]_a}(\{1,2\}) = w_{[E]_b,[d]_b}(\{1,2\}) = [6,8] \; .
\end{gather*}
In conclusion we cannot relate $\mathcal{F}(TAL;\cdot,\cdot)$ with
a game theoretic solution ot the interval bankruptcy game.

\end{xmpl}

The issue to provide an analog  of Theorem \ref{thm5} to this setting  remains on open question.

\end{document}